\documentclass{article}
\usepackage{graphicx} 
\usepackage{rotating} 
\usepackage{geometry}
\usepackage{tikzmacros}
\usepackage{subcaption}
\usetikzlibrary{decorations.pathreplacing,positioning}

\usepackage{amsfonts, amsmath, amssymb, amsthm}
\usepackage{bbm}

\usepackage{hyperref}
\usepackage[capitalize]{cleveref}

\usepackage{subcaption}

\newtheorem{definition}{Definition}[section]
\newtheorem{theorem}{Theorem}[section]
\newtheorem{lemma}{Lemma}[section]

\newcommand{\OPT}{\mathsf{OPT}}
\newcommand{\IB}{\mathsf{IB}}

\newcommand{\LB}{\mathsf{LB}}

\newcommand{\Greedy}{\textsf{Greedy}}
\newcommand{\Splay}{\textsf{Splay}}
\newcommand{\Tango}{\textsf{Tango}}
\newcommand{\MultiSplay}{\textsf{MultiSplay}}
\newcommand{\ChainSplay}{\textsf{ChainSplay}}
\newcommand{\GGreedy}{\textsf{GGreedy}}
\newcommand{\SmoothHeap}{\textsf{SmoothHeap}}
\newcommand{\Ancestors}{\mathsf{Ancestors}}
\newcommand{\Descendants}{\mathsf{Desc}}
\newcommand{\LeftSiblings}{\mathsf{LSiblings}}
\newcommand{\RightSiblings}{\mathsf{RSiblings}}
\newcommand{\Siblings}{\mathsf{Siblings}}

\newcommand{\ind}{\mathbbm{1}}
\newcommand{\ICost}{\mathsf{ICost}}
\newcommand{\GCost}{\mathsf{GCost}}
\newcommand{\Cost}{\mathsf{Cost}}
\newcommand{\Visible}{\mathsf{Visible}}

\newcommand{\Pref}{\mathsf{Pref}}
\newcommand{\LeftIntervals}{\mathsf{LIntervals}}
\newcommand{\RightIntervals}{\mathsf{RIntervals}}
\newcommand{\Intervals}{\mathsf{Intervals}}
\newcommand{\InvisiblePref}{\mathsf{InvisPref}}

\newcommand{\Patrascu}{P{\v{a}}tra\c{s}cu}

\title{The Greedy Binary Search Tree is Non-trivially Competitive\thanks{This work was conducted while the first and fourth authors were visiting University of Michigan, Ann Arbor.  Supported by NSF Grant CCF-2446604.}}

\author{Yuhao Guo\\{IIIS, Tsinghua University}
\and 
Seth Pettie\\
{University of Michigan}
\and
Daniel Skora\\
{University of Michigan}
\and
Chengzhang Wan\\
{IIIS, Tsinghua University}
}

\date{} 

\begin{document}

\maketitle

\begin{abstract}
    We prove that the $\Greedy$ binary search tree is 
    $2^{O(\sqrt{\log\log n})}$-competitive.  
    It is widely conjectured that $\Greedy$ is 
    $O(1)$-competitive,
    but before this work it was not known to be 
    $f$-competitive, for any non-trivial $f(n)=o(\log n)$.

\medskip 

    Our analysis differs from prior analyses of binary search trees.
    It takes what might be called a ``scaling'' approach, 
    where the cost at a refined scale is related to 
    the cost at a coarser scale, and Wilber's interleave lower bound.
\end{abstract}

\section{Introduction}

The \emph{dynamic optimality conjecture}~\cite{ST85} 
concerns the existence of an online dynamic binary 
search tree that is $O(1)$-competitive with the offline optimum.
All previous results on the dynamic optimality question can
be divided into two groups: 
\begin{description}
    \item[(1)] those that design a new dynamic 
    binary search tree with the explicit goal of proving
    $f(n)$-competitiveness, 
    for some non-trivial function $f(n) = o(\log n)$, and
    \item[(2)] those that try to analyze one of the two premier candidates for 
    dynamic optimality, namely $\Splay$~\cite{ST85} and $\Greedy$~\cite{Lucas,Munro00,DemaineGeometry}.
\end{description}

The most notable result in group (1) is certainly 
Demaine, Harmon, Iacono, and \Patrascu's~\cite{DemaineHIP07} $\Tango$ algorithm,
which is $O(\log\log n)$-competitive.  
Other results in this direction
are~\MultiSplay~\cite{WDS06} and \ChainSplay~\cite{Georgakopoulos08}, 
also $O(\log\log n)$-competitive, and \GGreedy~\cite{Chalermsook24}, 
which is $O(\sqrt{\log n})$-competitive.

\medskip
In group (2) the focus has been  analyzing $\Splay${} and $\Greedy$ on access
sequences that have some spatial or temporal locality, or on 
\emph{corollaries} of dynamic optimality on structured access sequences.
$\Splay$ and $\Greedy$ both satisfy the \emph{working set} 
bound~\cite{ST85,Fox11} (temporal locality) and the \emph{dynamic finger} bound~\cite{Cole00,ColeEtal00,IaconoL16} (spatial locality).
They both support \emph{sequential access} in optimal $O(n)$ time~\cite{Tarjan85,Sundar92,Elmasry04,Lucas,Fox11}, 
and
\emph{deque} sequences in nearly-linear time~\cite{Sundar92,Pettie08,Chalermsook23}.
Many access patterns can be characterized by a forbidden permutation,
e.g., sequential access is the unique permutation avoiding $(21)$, 
while deletion-only deque sequences avoid $(231)$ and $(213)$.
Both $\Splay$ and $\Greedy$ handle \emph{Postorder} sequences ($(312)$-avoiding) 
in $O(n)$ time~\cite{Chalermsook23,LevyT19}. 
\emph{Preorder} sequences (aka \emph{split}~\cite{Lucas88b} or \emph{traversal}~\cite{ST85} sequences, $(231)$-avoiding) seem more difficult.
$\Splay${} handles special cases in linear~\cite{ChaudhuriH90,LevyT19} or near-linear time $O(n\alpha(n))$~\cite{Lucas88b}, 
whereas $\Greedy$ handles all preorder sequences in $O(n2^{\alpha(n)})$ time.  
Following~\cite{ChalermsookGKMS15,Chalermsook23}, 
Chalermsook, Pettie, and Yingchareonthawornchai~\cite{Chalermsook24} 
proved that for any fixed permutation $\pi$, $\Greedy$ handles \emph{any} 
$\pi$-avoiding access sequence in $n2^{(1+o(1))\alpha(n)}$ time, 
where the $o(1)$ depends on $\pi$.

\medskip 
Very recently, and in independent work,
Chmel, Haeupler, Hlad\'{i}k, 
Kouck\'{y}, Roeyskoe, 
Rozho\v{n}, 
Sladk\'{y}, Tarjan~\cite{ChmelHHKRRST26}
proved that \Splay{} is $O(\log\log n(\log\log \log n)^2)$-competitive.
In this paper we prove the first  
competitiveness result for $\Greedy$.

\begin{theorem}\label{thm:intro-main-theorem}
    $\Greedy$ is $2^{O(\sqrt{\log\log n})}$-competitive.
\end{theorem}

Given the decades of research into dynamic optimality 
and its corollaries, 
the proof of \Cref{thm:intro-main-theorem} is 
unexpectedly simple.  
This is partly due to us using
the geometric characterization of Lucas's \Greedy{} algorithm~\cite{Lucas}
from 
Demaine, Harmon, Iacono, 
Kane, and \Patrascu~\cite{DemaineGeometry},
and partly due to a new style of analysis
that takes a ``scaling'' approach.  We define
several cost measures that operate at different scales,
and prove a couple lemmas 
relating two cost measures.

\subsection{Related Work}

One way to prove $f(n)$-competitiveness is to reason directly
about the relationship between the binary search tree under
consideration and an arbitrary (optimal) tree handling the same access sequence.  The Chmel et al.~\cite{ChmelHHKRRST26} result takes this approach, but most prior
work~\cite{DemaineHIP07,WDS06,Georgakopoulos08,Chalermsook24}
does not.  
Rather, they start from some
lower bound $\LB(\sigma)$ on the time needed to handle the particular 
access sequence $\sigma\in [n]^m$, $m\geq n$, 
then argue that the algorithm takes at most $f(n)\LB(\sigma)$ 
time.

Via a simple encoding argument, it is easy to show that, among the $n!$ permutations, at most $2^{O(c)}$ have complexity at most $c$.  In particular, nearly all permutations require $\Omega(n\log n)$ time, so static, balanced binary search trees are optimal on most access patterns.  Wilber~\cite{Wilber89}
established two intrinsic measures of hardness for access sequences, now commonly called \emph{Wilber I} or the \emph{interleave bound}, 
and \emph{Wilber II} or the \emph{funnel bound}.  As a consequence, 
Wilber demonstrated that the \emph{bit reversal} permutation 
is an example of a ``hard'' access sequence requiring 
$\Omega(n\log n)$ time.  Our competitiveness result, 
like~\cite{DemaineHIP07,WDS06,Georgakopoulos08,Chalermsook24},
uses the interleave bound.  
Lecomte and Weinstein~\cite{LecomteW20} proved that
the funnel bound is always as strong as the interleave bound,
and as much as a $\Theta(\log\log n)$-factor stronger on 
some access sequences.  Wilber's original definition of funnel 
treated time and space differently.  It is most natural
to express the access permutation\footnote{In the context of the dynamic optimality question, one may assume without loss of generality that the access sequence is a permutation.  See~\cite[Appendix E]{ChalermsookGKMS15-tr}.} as a
permutation matrix $P_\sigma \in \{0,1\}^{n\times n}$,
with one axis corresponding to time and the other to space.
Lecomte and Weinstein~\cite{LecomteW20} also reformulated
the funnel bound so that it is obviously indifferent to time and space, that is, the funnel bound of $P_\sigma$ is the same as $P_\sigma^{\top}$
and the reflections of $P_\sigma$ with time or space reversed.

In a 2005 technical report, 
Derryberry, Sleator, and Wang~\cite{DerryberrySW05}
gave a lower bound that subsumes the interleave and funnel bounds,
by considering how many axis-aligned boxes could be added to $P_\sigma$ 
according to certain rules.
In independent work, 
Demaine, Harmon, Iacono, Kane, and \Patrascu~\cite{DemaineGeometry}
\emph{characterized} valid binary search tree executions as
\emph{arboreally satisfied supersets} of $P_\sigma$.  
This provided a cleaner way to look at lower bounds, 
that subsumed~\cite{Wilber89,DerryberrySW05}, though it is still unknown
whether these absolute lower bounds are \emph{strictly} stronger than~\cite{DerryberrySW05} or Wilber's funnel bound~\cite{Wilber89}.
The arboreally satisfied superset perspective also makes it easier to talk 
about some binary search tree \emph{algorithms}.
Lucas's~\cite{Lucas} $\Greedy$ heuristic, for example, 
becomes trivial to describe when viewed from this perspective.

The $\Greedy$ algorithm, as Lucas~\cite{Lucas} originally described it in 1988
and Munro~\cite{Munro00} later rediscovered in 2000, 
restructures the access path to $\sigma(t)$ to optimally handle
\emph{future} accesses $\sigma(t+1),\sigma(t+2),\ldots$.  
Demaine et al.~\cite{DemaineHIP07} proved that it can be simulated online 
with an $O(1)$-factor slowdown.  Thus, in context we use \emph{$\Greedy$} 
to mean 
either the online or offline version.
Demaine et al.~\cite{DemaineGeometry} noted that on some access sequences $\sigma\in[n]^m$, $\Greedy$ is suboptimal by an additive $+m$. They reiterated Lucas's conjecture that $\Greedy$ is $O(1)$-competitive (dynamically optimal), 
and made a stronger conjecture that it is optimal up to an additive $+O(m)$.  
This stronger conjecture was later refuted by Sadeh and Kaplan~\cite{SadehK23},
who proved that on some access sequences, $\Greedy$ is suboptimal by an 
additive $+\Omega(m\log\log n)$, and that it cannot be (multiplicatively) 
$c$-competitive for any $c<2$.

Kozma and Saranurak~\cite{KozmaS20}
introduced a self-adjusting priority queue called the \emph{\SmoothHeap}, and proved that it is ``dual'' to the \Greedy{} BST in the following sense. When the \SmoothHeap{} is used to 
sort a permutation\footnote{\textsf{Insert} the $n$ elements then perform $n$ \textsf{DeleteMin}s.} 
$\sigma : [n]\to [n]$, the cost and pattern of elements touched is \emph{identical} to sorting the inverse permutation $\sigma^{-1}$ with \Greedy.\footnote{\textsf{Insert} the elements into an empty BST one-by-one.}
The optimum BST cost of $P_\sigma$ is invariant under rotation/reflection, so all of our results on \Greedy{} sorting apply to \SmoothHeap{} sorting as well.  The \SmoothHeap{} is \emph{naturally} an online algorithm, which makes it more attractive than \Greedy{} for sorting.

\subsection{Organization}

\Cref{sect:preliminaries} reviews some
well known results, such as Demaine et al.'s~\cite{DemaineGeometry}
characterization of binary search tree executions
as arboreally satisfied supersets, 
what the \Greedy{} algorithm looks like in this framework, and Wilber's~\cite{Wilber89} interleave lower bound, then proceeds to define a family
of references trees and a series of 
associated cost measures.
\Cref{sect:competitive-analysis-of-greedy}
proves two lemmas about the relationship between
different cost measures, which implies \Cref{thm:intro-main-theorem}.
We conclude with some remarks in \Cref{sect:conclusion}.

\section{Preliminaries}\label{sect:preliminaries}

We now describe the geometric view of binary search trees, and of $\Greedy$ in particular, 
due to Demaine, Harmon, Iacono, Kane, and \Patrascu~\cite{DemaineGeometry}.  
Consider a point set $P\subset \mathbb{Z}\times\mathbb{Z}$.  For any two points $p_1=(x_1, y_1)$ and $p_2=(x_2, y_2)$ in $P$, we say $p_1$ and $p_2$ are \emph{arboreally satisfied} with respect to $P$ if either (1) $p_1$ and $p_2$ lie on the same horizontal or vertical line or (2) the rectangle $[x_1, x_2]\times [y_1, y_2]$ contains some point of $P\setminus \{p_1, p_2\}$.  We say the point set $P$ is arboreally satisfied if all pairs of points in $P$ are arboreally satisfied with respect to $P$.

Demaine et al.~\cite{DemaineGeometry} showed that every BST execution maps onto an arboreally satisfied point set, and every such point set is mapped to by at least one BST execution.  In particular, the access $\sigma(i)$ causes the BST to reshape a connected subtree $\tau(i)$ containing $\sigma(i)$ and the root.
The execution is represented as a point-set $P'$ where for each 
$j\in\tau(i)$, $(j,i)\in P'$.
$P'$ is arboreally satisfied.
Thus, the optimum cost of accessing $\sigma$ is the minimum cardinality
arboreally satisfied superset $P'\supset P_\sigma = \{(\sigma(i),i) \mid i\in [m]\}$.

\begin{figure}[h]
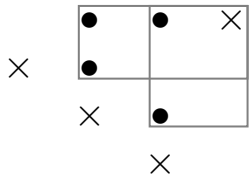

\centering
\scalebox{1.5}{$
\begin{NiceMatrix}[name=A]
 & \bullet & \bullet & \times \\
 \times & \bullet & & \\
 & \times & \bullet & \\
 & & \times & \\
 \CodeAfter
 \tikz{
    \draw[gray, line width=.5pt] (A-1-4.north east) rectangle (A-2-2.south west);
    \draw[gray, line width=.5pt] (A-1-4.north east) rectangle (A-3-3.south west);
}
\end{NiceMatrix}
$}
\caption{The geometric view of $\Greedy$ for the access sequence $\sigma=[3,2,1,4]$.  The two rectangles have in their lower left corners the two ``unsatisfied points'' with respect to the point $(4,4)$ representing the search at time $4$.}
\label{fig:greedy_ortho}
\end{figure}

The most natural ASS algorithm 
is the \Greedy{} line sweep algorithm. Given an input set 
$P_\sigma = \{(\sigma(i),i)\}$, 
we initialize $P'\gets P_\sigma$,
and process $P_\sigma$ row-by-row.
At row (time) $t$,
consider any point $(i',t')\in P'$ 
for $t'<t$.  
If the closed rectangle with corners at $(i',t')$ and $(\sigma(t), t)$ does not contain any points besides the corners themselves, then 
set $P' \gets P' \cup \{(i', t)\}$.
For each $(\sigma(t),t) \in P_\sigma$, 
we say $\sigma(t)$ was \emph{accessed} at time $t$, and for $(i',t)\in P'$ 
we say item $i'$ was \emph{touched} at time $t$.  Accessed points are depicted by $\times$, and all other points by $\bullet$.  An example is shown in \Cref{fig:greedy_ortho}.

If the \Greedy{} line-sweep algorithm is run directly on $P_\sigma$, this corresponds to picking an \emph{ideal} initial tree $T_0$.  To encode a specific tree $T_0$, we can run \Greedy{} on the point-set 
$P_\sigma \cup P_{T_0}$, where 
\[
P_{T_0} = \{(x, -\operatorname{depth}_{T_0}(x))  \mid x\in [n]\}.
\]
An example is shown in \Cref{fig:greedy-balanced-tree}.

The \Greedy{} line-sweep algorithm touches precisely the same elements~\cite{DemaineGeometry} 
as the offline \Greedy{} algorithm described by Lucas~\cite{Lucas} and Munro~\cite{Munro00}.
Demaine et al.~\cite{DemaineGeometry} proved that it 
could be implemented online, without knowledge of the 
future accesses, with an $O(1)$ factor slowdown. 
(Although the \emph{line-sweep} algorithm is online, it only encodes which elements are touched at a particular time, not how the tree is rearranged.  The ideal rearrangement is unknowable at the time of the access as it depends on the pattern of \emph{future} accesses.  Nonetheless, an $O(1)$-approximation to the ideal tree can be simulated online. 
For this reason, \Greedy{} may be represented as a multiway tree, with the binary representation of the elements in each node left indeterminate.  See~\Cref{fig:compare-views}.)

\begin{figure}[t]
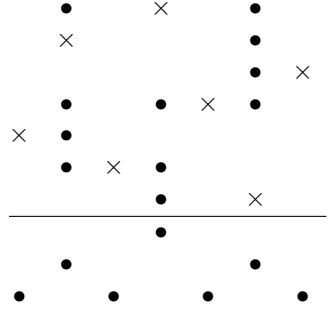

\centering
$\begin{NiceMatrix}

&\bullet&&\times&&\bullet&\\
&\times&&&&\bullet&\\
&&&&&\bullet&\times\\
&\bullet&&\bullet&\times&\bullet&\\
\times&\bullet&&&&&\\
&\bullet&\times&\bullet&&&\\
&&&\bullet&&\times&\\
\hline
&&&\bullet&&&\\
&\bullet&&&&\bullet&\\
\bullet&&\bullet&&\bullet&&\bullet\\

\end{NiceMatrix}$
\caption{The geometric view of $\Greedy$ for some search sequence $\sigma$ on a balanced initial tree.}
\label{fig:greedy-balanced-tree}
\end{figure}

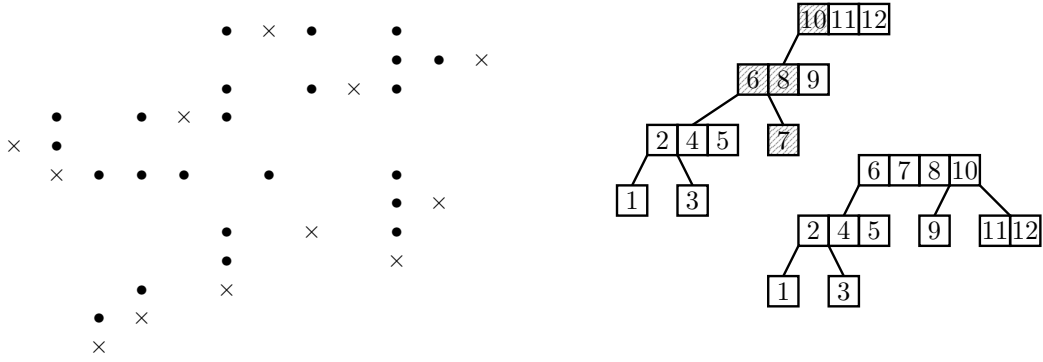
\begin{figure}[ht]
\centering
\begin{subfigure}[b]{0.48\textwidth}
\centering
\vspace{0pt}
\scalebox{0.9}{$
\begin{NiceMatrix}
& & & & & \bullet & \times & \bullet & & \bullet & & \\
& & & & & & & & & \bullet & \bullet & \times \\
& & & & & \bullet & & \bullet & \times & \bullet & & \\
& \bullet & & \bullet & \times & \bullet & & & & & & \\
\times & \bullet & & & & & & & & & & \\
& \times & \bullet & \bullet & \bullet & & \bullet & & & \bullet & & \\
& & & & & & & & & \bullet & \times & \\
& & & & & \bullet & & \times & & \bullet & & \\
& & & & & \bullet & & & & \times & & \\
& & & \bullet & & \times & & & & & & \\
& & \bullet & \times & & & & & & & & \\
& & \times & & & & & & & & & \\
\end{NiceMatrix}
$}
\caption{The geometric view of the execution of $\Greedy$, with no explicit initial tree $T_0$.}
\label{fig:temp-matrix}
\end{subfigure}
\hfill
\begin{subfigure}[b]{0.48\textwidth}
\centering
\vspace{0pt}
\begin{tikzpicture}[x=1cm,y=1cm,line width=0.9pt, scale=0.8]

  \SquareAt{1}{-0.5}{-1}[12]
  \SquareAt{2}{-1}{-1}[11]
  \ShadedSquareAt{3}{-1.5}{-1}[10]

  \SquareAt{4}{-1.5}{-2}[9]
  \ShadedSquareAt{5}{-2}{-2}[8]
  \ShadedSquareAt{7}{-2.5}{-2}[6]
    
  \ShadedSquareAt{6}{-2}{-3}[7]
  \SquareAt{8}{-3}{-3}[5]
  \SquareAt{9}{-3.5}{-3}[4]
  \SquareAt{11}{-4}{-3}[2]
    
  \SquareAt{10}{-3.5}{-4}[3]
  \SquareAt{12}{-4.5}{-4}[1]

  \ConnectChild{3}{5}{L}
  \ConnectChild{7}{9}{L}
  \ConnectChild{7}{6}{R}
  \ConnectChild{11}{12}{L}
  \ConnectChild{11}{10}{R}

  \SquareAt{15}{1}{-3.5}[10]
  \SquareAt{17}{0.5}{-3.5}[8]
  \SquareAt{18}{0}{-3.5}[7]
  \SquareAt{19}{-0.5}{-3.5}[6]

  \SquareAt{20}{-0.5}{-4.5}[5]
  \SquareAt{21}{-1}{-4.5}[4]
  \SquareAt{23}{-1.5}{-4.5}[2]

  \SquareAt{13}{2}{-4.5}[12]
  \SquareAt{14}{1.5}{-4.5}[11]
  \SquareAt{16}{0.5}{-4.5}[9]
    
  \SquareAt{22}{-1}{-5.5}[3]
  \SquareAt{24}{-2}{-5.5}[1]

  \draw (1.75, -4.25) -- (1.25, -3.75);
  \ConnectChild{19}{21}{L}
  \ConnectChild{15}{16}{L}
  \ConnectChild{23}{24}{L}
  \ConnectChild{23}{22}{R}
\end{tikzpicture}
\vspace*{.7cm}

\caption{The corresponding tree-based representation, shown before and after $7$ is accessed.}
\label{fig:temp-tree}
\end{subfigure}
\caption{Two equivalent views of $\Greedy$.}
\label{fig:compare-views}
\end{figure}

\medskip 

Define $\OPT(\sigma,T_0)$ to be the minimum cardinality
of an arboreally satisfied superset of 
$P_\sigma \cup P_{T_0}$ in the positive quadrant,
and $\Cost(\sigma,T_0)$ to be the size of the arboreally satisfied superset selected by \Greedy.
Our goal is to bound the competitive ratio $\Cost(\sigma,T_0)/\OPT(\sigma,T_0)$.
The initial tree does not play a big role in our analysis,
so we often omit it from this notation and write 
$\OPT(\sigma),\Cost(\sigma)$.

\subsection{The Interleave Bound}

Wilber's \cite{Wilber89} \emph{interleave bound} $\IB(\sigma)$ lower bounds $\OPT(\sigma,T_0)$ for any $\sigma\in[n]^m$ and $T_0$.  It is defined with respect to an arbitrary \emph{reference tree} $R$ on leaf-set $[n]$.\footnote{It can also be defined w.r.t.~an $R$ whose node-set is $[n]$; see~\cite{DemaineHIP07}.}
At time 0 each internal node of $R$ 
picks its first child to be its \emph{preferred child}.  
Let $a$ be an ancestor of (leaf) $\sigma(t)$, and $b$ be the child of $a$ that is an ancestor of $\sigma(t)$.  After access $\sigma(t)$, $b$ becomes the preferred child of $a$. 
Define $\Lambda(\sigma,R,t)$ to be the number of ancestors of $\sigma(t)$ that \emph{change} their preferred child in response to access $\sigma(t)$.  
The interleave bound with respect to $R$ is:
\[
\IB(\sigma, R) = \sum_{t\in [m]} \Lambda(\sigma, R, t).
\]

\begin{theorem}[Wilber~\cite{Wilber89}]
Fix an access sequence $\sigma\in[n]^m$.  
For any reference tree $R$ and initial tree $T_0$,
\[
\OPT(\sigma,T_0) \geq \Omega(\IB(\sigma, R) + m).
\]
\end{theorem}

The interleave bound is known to be asymptotically not tight, even under an optimal choice of reference tree~\cite{LecomteW20}.


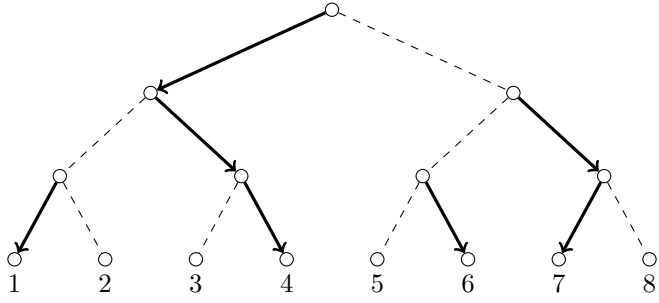
\begin{figure}[t]
\centering
\begin{tikzpicture}[
    every node/.style={circle, draw, inner sep=1.5pt, minimum size=5pt}
]

\node (r) at (0,0) {};
\node (a) at (-2.4,-1.1) {};
\node (d) at (2.4,-1.1) {};

\node (b) at (-3.6,-2.2) {};
\node (c) at (-1.2,-2.2) {};
\node (e) at (1.2,-2.2) {};
\node (f) at (3.6,-2.2) {};

\node[label=below:{$1$}] (l1) at (-4.2,-3.3) {};
\node[label=below:{$2$}] (l2) at (-3.0,-3.3) {};
\node[label=below:{$3$}] (l3) at (-1.8,-3.3) {};
\node[label=below:{$4$}] (l4) at (-0.6,-3.3) {};
\node[label=below:{$5$}] (l5) at (0.6,-3.3) {};
\node[label=below:{$6$}] (l6) at (1.8,-3.3) {};
\node[label=below:{$7$}] (l7) at (3.0,-3.3) {};
\node[label=below:{$8$}] (l8) at (4.2,-3.3) {};

\draw[dashed] (r) -- (d);

\draw[dashed] (d) -- (e);
\draw[dashed] (d) -- (f);
\draw[dashed] (a) -- (c);
\draw[dashed] (a) -- (b);

\draw[dashed] (b) -- (l1);
\draw[dashed] (b) -- (l2);
\draw[dashed] (c) -- (l3);
\draw[dashed] (c) -- (l4);

\draw[dashed] (e) -- (l5);
\draw[dashed] (e) -- (l6);
\draw[dashed] (f) -- (l7);
\draw[dashed] (f) -- (l8);

\draw[->, very thick] (r) -- (a);
\draw[->, very thick] (a) -- (c);
\draw[->, very thick] (d) -- (f);

\draw[->, very thick] (b) -- (l1);
\draw[->, very thick] (c) -- (l4);
\draw[->, very thick] (e) -- (l6);
\draw[->, very thick] (f) -- (l7);

\end{tikzpicture}
\caption{A visualization of the preferred child pointers for a particular reference tree following the execution of some access sequence $\sigma$.  The last item to be accessed was $4$.}
\label{fig:interleave}
\end{figure}

\subsection{Reference Trees and Cost Measures}

Assume without loss of generality that $n=2^{2^k}$ for some $k=p^2$ and integer $p=\sqrt{\log\log n}$.  We consider a collection of reference trees $\{R_i\}$ for $0\le i\le k$.  Each $R_i$ is a rooted, depth $2^{k-i}$ tree with branching factor $2^{2^i}$ and $n$ leaves labeled by the elements $[n]$.  In particular, the reference tree $R_0$ is a rooted binary tree with depth $2^k$ and $n$ leaves, and each $R_i$ is obtained by ``compressing'' adjacent layers of $R_{i-1}$ to halve its depth.  We consider nodes in different reference trees to be the same if they have the 
same leaf-descendants, so 
$V(R_0) \subset V(R_1) \subset \cdots \subset V(R_k)$.

Let $0\le i\le j\le k$, and $x$ be an arbitrary node in $R_i$. The following notation will be used throughout the proof.

\begin{itemize}
    \item Define $\Ancestors_i(x)$ to be the set of ancestors of $x$ in $R_i$, including $x$ itself.
    
    \item Define $\LeftSiblings_i(x)$ to be the set of siblings of $x$ in $R_i$ strictly to the left to $x$, and define $\RightSiblings_i(x)$ similarly. Define $\Siblings_i(x)=\LeftSiblings_i(x)\cup \RightSiblings_i(x)$.
    
    \item Define $\Descendants_j(x)$ to be the set of nearest descendants of $x$ in $R_i$ that also belong
    to $R_j$, and define $\Descendants_j^*(x)$ in the same way with respect to \emph{strict} descendants of $x$.
    In other words, $\Descendants_i(x)=\{x\}$, $\Descendants^*_i(x)$ is the set of children of $x$ in $R_i$, and $\Descendants_k(x)$ are the leaf-descendants of $x$ in $R_i$, or in fact any reference tree containing $x$, if $x$ is not the root.  We treat $\Descendants_k(x)$ as an \emph{interval} of $[n]$.

    \item Define $\Pref_{j,t}(x)\in \Descendants^*_j(x)$ to be the unique element of $\Descendants^*_j(x)$ such that the path from $x$ to $\Pref_{j,t}(x)$ is composed of preferred-child edges at time $t$, that is, 
    the most recently accessed element at or before time $t$ in the subtree of $x$ also lies in the 
    subtree of $\Pref_{j,t}(x)$.
\end{itemize}

An example illustrating the notation is shown in \Cref{fig:notation}.

\begin{figure}[t]
\centering
\begin{tikzpicture}[
    scale=0.85,
    transform shape,
    every node/.style={circle, draw, inner sep=1.5pt, minimum size=5mm}
]

\node (A) at (0,0) {$A$};
\node (B) at (-4.0,-1.1) {$B$};
\node (C) at (4.0,-1.1) {$C$};

\node (D) at (-6.0,-2.2) {$D$};
\node (E) at (-2.0,-2.2) {$E$};
\node (F) at (2.0,-2.2) {$F$};
\node (G) at (6.0,-2.2) {$G$};

\node (H) at (-7.0,-3.3) {$H$};
\node (I) at (-5.0,-3.3) {$I$};
\node (J) at (-3.0,-3.3) {$J$};
\node (K) at (-1.0,-3.3) {$K$};
\node (L) at (1.0,-3.3) {$L$};
\node (M) at (3.0,-3.3) {$M$};
\node (N) at (5.0,-3.3) {$N$};
\node (O) at (7.0,-3.3) {$O$};

\node (1) at (-7.5,-4.4) {$1$};
\node (2) at (-6.5,-4.4) {$2$};
\node (3) at (-5.5,-4.4) {$3$};
\node (4) at (-4.5,-4.4) {$4$};
\node (5) at (-3.5,-4.4) {$5$};
\node (6) at (-2.5,-4.4) {$6$};
\node (7) at (-1.5,-4.4) {$7$};
\node (8) at (-0.5,-4.4) {$8$};
\node (9) at (0.5,-4.4) {$9$};
\node (10) at (1.5,-4.4) {$10$};
\node (11) at (2.5,-4.4) {$11$};
\node (12) at (3.5,-4.4) {$12$};
\node (13) at (4.5,-4.4) {$13$};
\node (14) at (5.5,-4.4) {$14$};
\node (15) at (6.5,-4.4) {$15$};
\node (16) at (7.5,-4.4) {$16$};

\draw[dashed] (A) -- (C);
\draw[->, very thick] (A) -- (B);
\draw[dashed] (B) -- (D);
\draw[->, very thick] (B) -- (E);
\draw[dashed] (C) -- (G);
\draw[->, very thick] (C) -- (F);
\draw[dashed] (D) -- (H);
\draw[->, very thick] (D) -- (I);
\draw[dashed] (E) -- (K);
\draw[->, very thick] (E) -- (J);
\draw[dashed] (F) -- (M);
\draw[->, very thick] (F) -- (L);
\draw[dashed] (G) -- (N);
\draw[->, very thick] (G) -- (O);
\draw[dashed] (H) -- (1);
\draw[->, very thick] (H) -- (2);
\draw[dashed] (I) -- (4);
\draw[->, very thick] (I) -- (3);
\draw[dashed] (J) -- (6);
\draw[->, very thick] (J) -- (5);
\draw[dashed] (K) -- (7);
\draw[->, very thick] (K) -- (8);
\draw[dashed] (L) -- (9);
\draw[->, very thick] (L) -- (10);
\draw[dashed] (M) -- (11);
\draw[->, very thick] (M) -- (12);
\draw[dashed] (N) -- (14);
\draw[->, very thick] (N) -- (13);
\draw[dashed] (O) -- (15);
\draw[->, very thick] (O) -- (16);

\node[
    draw=none,
    anchor=north west,
    align=left,
] at (-7.0,-4.8) {
The drawn tree is $R_0$ at time $t$.\\
$V(R_1)=\{A,D,E,F,G,1,2,\ldots,16\}$.\\
$V(R_2)=\{A,1,2,\ldots,16\}$.\\
$\Ancestors_0(7)=\{A,B,E,K,7\}$.\\
$\Ancestors_1(7)=\{A,E,7\}$.\\
$\Siblings_0(E)=\{D\}$.\\
$\Siblings_1(E)=\{D,F,G\}$.
};

\node[
    draw=none,
    anchor=north west,
    align=left,
] at (1.0,-5.2) {
$\Descendants_0(E)=\Descendants_1(E)=\{E\}$.\\
$\Descendants^*_1(E)=\Descendants_2(E)=\{5,6,7,8\}$.\\
$\Pref_{0,t}(A)=B$.\\
$\Pref_{1,t}(A)=E$.
};

\end{tikzpicture}
\caption{Illustration of the notation when $k=2$.}
\label{fig:notation}
\end{figure}


Suppose we run $\Greedy$ on an access sequence $\sigma \in [n]^m$ and an arbitrary initial tree $T_0$.  For a time $t\in [m]$ and set $S\subseteq [n]$, let $\ind(t,S)\in\{0,1\}$ be an indicator for whether $\Greedy$ touches \emph{any} element in the set $S$ while accessing $\sigma(t)$ at time $t$.  For each reference tree $R_i$, we introduce two cost measures for a search at time $t$: the \emph{individual cost} $\ICost_i(t)$ and the \emph{group cost} $\GCost_i(t)$.


\begin{align*}
    \ICost_i(t) &= \sum_{y\in \Ancestors_i(\sigma(t))} \sum_{z \in \Siblings_i(y)} \ind(t,\Descendants_k(z)).\\
    \GCost_i(t) &= \sum_{y \in \Ancestors_i(\sigma(t))} \left(\ind\left(t,\bigcup_{z \in \LeftSiblings_i(y)} \Descendants_k(z)\right) + \ind\left(t,\bigcup_{z \in \RightSiblings_i(y)} \Descendants_k(z)\right)\right).
\end{align*}

Let us state these two cost measures in plain English.  For each reference tree $R_i$ and time $t$, we can partition $[n]$ into a singleton set $\{\sigma(t)\}$ and several intervals --- $2^{k-i}(2^{2^i}-1)$ to be exact --- which are the siblings of ancestors of $\sigma(t)$ in $R_i$.  The individual cost measures how many of these intervals are touched when accessing $\sigma(t)$.  In contrast, the group cost coalesces these intervals into at most $2\cdot 2^{k-i}$ 
groups, two for each ancestor, then measures how many such groups were touched by the access.  See \Cref{fig:ICost-GCost}.

\begin{sidewaysfigure}
  \centering
  \includegraphics[width=0.98\textheight]{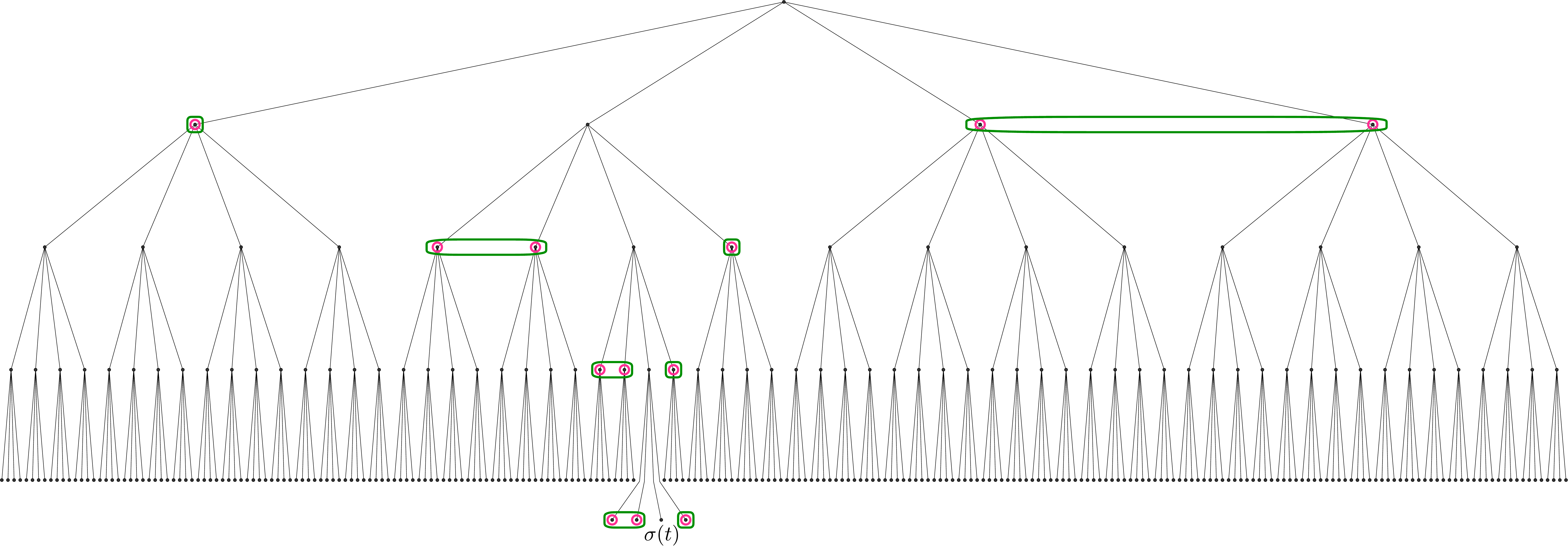}
  \caption{The reference tree $R_1$ with $2^{2^3}$ leaves and branching factor $2^{2^1}=4$.  The $\ICost$ of accessing $\sigma(t)$ is the number of intervals touched that correspond to siblings of ancestors of $\sigma(t)$. These are marked with 12 pink circles.  The $\GCost$ measures the same thing, but with left siblings and right siblings of each ancestor grouped into a single interval.  These are indicated by 8 green ovals.}
  \label{fig:ICost-GCost}
\end{sidewaysfigure}

Define $\ICost_i = \sum_{t\in [m]} \ICost_i(t)$ 
and $\GCost_i = \sum_{t\in [m]} \GCost_i(t)$.  The following lemma provides some simple relations on the various auxiliary cost measures and the actual \Greedy{} 
cost $\Cost(\sigma)$.

\begin{lemma}\label{lem:ICost-GCost-properties}
    For any access sequence $\sigma\in [n]^m$, $\ICost$ and $\GCost$ have the following properties.
    \begin{enumerate}
        \item $\GCost_k \leq \GCost_{k-1} \leq \cdots \GCost_0 = \ICost_0 \leq \ICost_1 \leq \cdots \ICost_k$.
        \item $\Cost(\sigma) = \ICost_k + m$.
        \item $\GCost_k \leq 2m$.
    \end{enumerate}
\end{lemma}

\begin{proof}
The inequality chains $\GCost_k \leq \cdots \leq \GCost_0$  and $\ICost_0 \leq \cdots \leq \ICost_k$ follow immediately from the definitions of $\GCost$ and $\ICost$.  We have that $\GCost_0 = \ICost_0$ because $R_0$ is binary and hence $\Siblings(y)$ is a singleton set for each ancestor $y$ of $\sigma(t)$.  This proves claim (1).

Since $R_k$ is a star with $n$ leaves, $\ICost_k(t)$ measures the total
number of elements touched when accessing $\sigma(t)$, 
excluding $\sigma(t)$ itself, which is always touched.  Hence, $\Cost(\sigma)$ is exactly $m$ more than $\ICost_k$.  Similarly, $\GCost_k(t) \leq 2m$ since only two groups are considered at each time $t$, namely $\LeftSiblings_k(\sigma(t))$ and $\RightSiblings_k(\sigma(t))$.
\end{proof}

\subsection{Visible Sets}

For $x\in [n]$, denote by $h_t(x)$ the height ($y$-coordinate) of the element $x$ in the geometric view of $\Greedy$ after accessing $\sigma(t)$, that is,
the maximum $i$ for which $(x,i)$ is in the point-set.
For any interval $I\subseteq [n]$, denote $h_t(I):=\max_{x\in I}(h_t(x))$.

Fix an interval $I\subseteq[n]$ 
and let $S=\{I_1,\ldots,I_p\}$ be a partition of $I$ 
into disjoint intervals,
sorted in increasing order.

\begin{definition}
    The \emph{visible set} of $S=\{I_1,\ldots,I_p\}$ at time $t$, denoted $\Visible_t(S)$, is a subset of $S$ containing every interval that has a \emph{strict} prefix maximal or suffix maximal $h_t(\cdot)$ value.  Formally, $I_j\in \Visible_t(S)$ if $h_t(I_k)<h_t(I_j)$ holds for all $k<j$ or for all $k>j$.
\end{definition}

\Cref{prop:hull} shows that accessing an element in $I$ increases the visibility set by a constant, and accessing an element outside $I$ can only reduce the size of the visibility set.

\begin{lemma} \label{prop:hull}
    Let $S=\{I_1,\ldots,I_p\}$ be a partition of an interval $I\subseteq[n]$ into disjoint intervals.

    \begin{enumerate}
        \item If $\sigma(t)\in I$, then $|\Visible_t(S)\setminus \Visible_{t-1}(S)|\le 2$.
        \item If $\sigma(t)\notin I$, then $\Visible_t(S)\subseteq \Visible_{t-1}(S)$. If $r$ intervals in $S$ are touched in this access, we have $|\Visible_{t-1}(S)\setminus \Visible_{t}(S)|\ge r-2$.
    \end{enumerate}
\end{lemma}

\begin{proof}

    Let $\{I_{j_1}, \dots, I_{j_r}\} \subseteq S$ be the set of intervals touched by the access at time $t$, where $j_1 < j_2 < \dots < j_r$. Since $h_t(I_{j_1})=h_t(I_{j_2})=\cdots=h_t(I_{j_r})$, we know that $I_{j_2},\ldots,I_{j_{r-1}}$ are not in the visible set. Thus, only $I_{j_1}$ and $I_{j_r}$ can become visible after accessing $\sigma(t)$, which means $|\Visible_t(S)\setminus \Visible_{t-1}(S)|\le 2$.

    If $\sigma(t)\notin I$, only visible intervals can be touched. Specifically, only prefix maxima can be touched if $\sigma(t)$ is smaller than elements in $I$; only suffix maxima can be touched if $\sigma(t)$ is larger than elements in $I$.  In other words, $\{I_{j_1},\ldots,I_{j_r}\}\subseteq\Visible_{t-1}(S)$. Thus, no invisible intervals can become visible after this access. Moreover, originally visible intervals $I_{j_2},\ldots,I_{j_{r-1}}$ become invisible after this access, so we have $|\Visible_{t-1}(S)\setminus \Visible_{t}(S)|\ge r-2$.
\end{proof}

\medskip 
The examples of the two cases of \Cref{prop:hull} 
are illustrated in \Cref{fig:prop-hull-1,fig:prop-hull-2}.


\begin{figure}[ht]
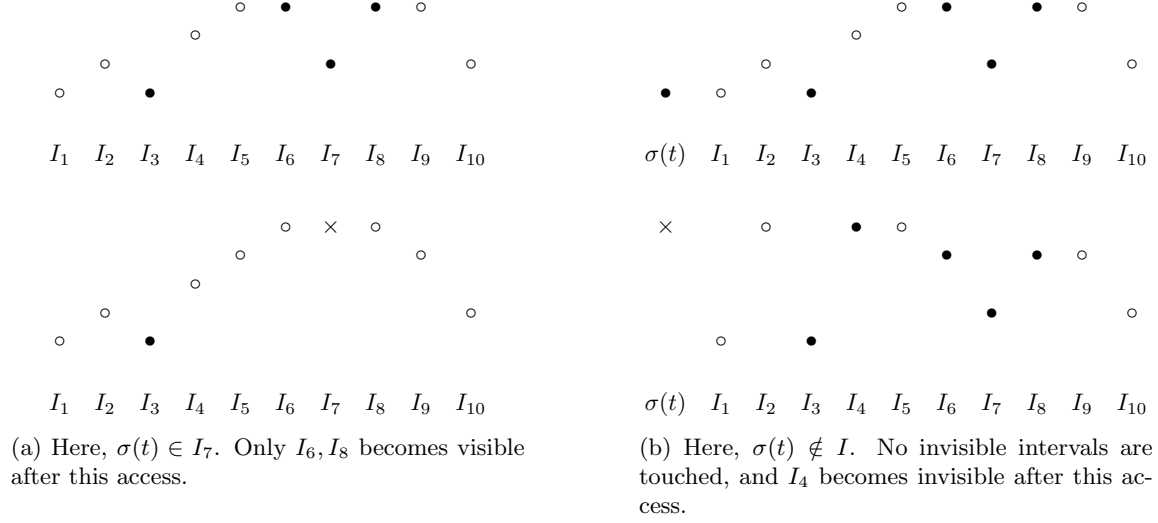

    \centering

    \begin{subfigure}[t]{0.45\textwidth}
        \centering

        \scalebox{0.9}{$
        \begin{NiceMatrix}
         & & & & {\circ} & \bullet & & \bullet & {\circ} & \\
         & & & {\circ} & & & & & & \\
         & {\circ} & & & & & \bullet & & & {\circ} \\
         {\circ} & & \bullet & & & & & & & \\
         \noalign{\vskip 5mm}
         I_1 & I_2 & I_3 & I_4 & I_5 & I_6 & I_7 & I_8 & I_9 & I_{10}
        \end{NiceMatrix}
        $}

        \vspace{1.5em}

        \scalebox{0.9}{$
        \begin{NiceMatrix}
         & & & & & {\circ} & \times & {\circ} & & \\
         & & & & {\circ} & & & & {\circ} & \\
         & & & {\circ} & & & & & & \\
         & {\circ} & & & & & & & & {\circ} \\
         {\circ} & & \bullet & & & & & & & \\
         \noalign{\vskip 5mm}
         I_1 & I_2 & I_3 & I_4 & I_5 & I_6 & I_7 & I_8 & I_9 & I_{10}
        \end{NiceMatrix}
        $}

        \caption{Here, $\sigma(t)\in I_7$.  Only $I_6,I_8$ becomes visible after this access.}
        \label{fig:prop-hull-1}
    \end{subfigure}
    \hfill
    \begin{subfigure}[t]{0.45\textwidth}
        \centering
        \scalebox{0.9}{$
        \begin{NiceMatrix}
         & & & & & {\circ} & \bullet & & \bullet & {\circ} & \\
         & & & & {\circ} & & & & & & \\
         & & {\circ} & & & & & \bullet & & & {\circ} \\
         \bullet & {\circ} & & \bullet & & & & & & & \\
         \noalign{\vskip 5mm}
         \sigma(t) & I_1 & I_2 & I_3 & I_4 & I_5 & I_6 & I_7 & I_8 & I_9 & I_{10}
        \end{NiceMatrix}
        $}

        \vspace{1.5em}

        \scalebox{0.9}{$
        \begin{NiceMatrix}
         \times & & {\circ} & & \bullet & {\circ} & & & & & \\
         & & & & & & \bullet & & \bullet & {\circ} & \\
         & & & & & & & & & & \\
         & & & & & & & \bullet & & & {\circ} \\
         & {\circ} & & \bullet & & & & & & & \\
         \noalign{\vskip 5mm}
         \sigma(t) & I_1 & I_2 & I_3 & I_4 & I_5 & I_6 & I_7 & I_8 & I_9 & I_{10}
        \end{NiceMatrix}
        $}

        \caption{Here, $\sigma(t)\notin I$.  No invisible intervals are touched, and $I_4$ becomes invisible after this access.}
        \label{fig:prop-hull-2}
    \end{subfigure}

    \caption{The geometric view before and after some $\sigma(t)$ is accessed.  We compact each interval $I_j$ and only keep the highest point within each interval.  Intervals in the visible set are shown as hollow bullets.}
    \label{fig:visible-sets}
\end{figure}



\section{Competitive Analysis of \Greedy}\label{sect:competitive-analysis-of-greedy}

Recall from \Cref{lem:ICost-GCost-properties} that the 
cost of handling $\sigma\in[n]^m$ with \Greedy{} is equal to $\ICost_k+m$.  
We will bound $\ICost_k$ with respect to the interleave bound $\IB=\IB(\sigma, R_0)$.

\begin{theorem}\label{thm:main}
    $\ICost_k\le 2^{O\left(\sqrt {\log\log n}\right)}\cdot(\IB(\sigma,R_0)+n+m)$.
\end{theorem}

As $\IB(\sigma,R_0)$ is a lower bound on $\OPT(\sigma)$, \Cref{thm:main} implies that \Greedy{} is 
$2^{O\left(\sqrt{\log\log n}\right)}$-competitive 
on sequences of length $m=\Omega(n)$, proving
\Cref{thm:intro-main-theorem}.
The proof of \Cref{thm:main} follows easily from the following two technical lemmas relating $\ICost$s and $\GCost$s.  
Together, the lemmas provide an unbroken chain 
connecting the true cost $\ICost_k$ back to the trivial cost $\GCost_k$ and the interleave bound.

\begin{lemma}\label{lem:v}
    For $0\le i<j\le k$, we have

    $$\GCost_i\le 2\cdot \GCost_j+2^{j-i}\cdot\left(2\cdot\IB(\sigma,R_j)+4n\right).$$
\end{lemma}

\begin{lemma}\label{lem:w}
    For $0\le i<j\le k$, we have

    $$\ICost_j\le 2\cdot \ICost_i+2^{j-i}\cdot\left(2\cdot \IB(\sigma,R_i) +2n+2\cdot \GCost_j\right).$$
\end{lemma}

In \Cref{sec:l1,sec:l2} we prove \Cref{lem:v,lem:w}, respectively, and in \Cref{sec:main_thm}, we complete the proof of \Cref{thm:main}.

\subsection{Proof of \Cref{lem:v}}\label{sec:l1}

Fix reference trees $R_i$ and $R_j$ with $i<j$.  
We will use a potential argument to prove the lemma. 
The total increase of potential is upper bounded by a combination of $n$ and the interleave bound, while the total decrease of potential is lower bounded by 
$\GCost_i-2\cdot \GCost_j$.

Fix a time $t$ and a non-leaf node $y\in V(R_j)$.
Let $y=y_0\to y_1\to\cdots\to y_{2^{j-i}}=\Pref_{j,t}(y)$
be the path in $R_i$ from $y$ to its 
preferred child in $R_j$.
The interval sets 
$\LeftIntervals,\RightIntervals$ 
are defined with respect to $y$ and $\Pref_{j,t}(y)$.

$$\begin{aligned}
    \LeftIntervals_t(y) &=\left\{\bigcup_{y'\in\LeftSiblings_i(y_q)}\Descendants_k(y')\;\:\middle|\:\; 1\le q\le 2^{j-i} \right\},\\
    \RightIntervals_t(y) &=\left\{\bigcup_{y'\in\RightSiblings_i(y_q)}\Descendants_k(y') \;\:\middle|\:\; 1\le q\le 2^{j-i} \right\}.
\end{aligned}$$

$\LeftIntervals_t(y)$ and $\RightIntervals_t(y)$ are partitions of the leaf-descendants of $y$ that are 
strictly to the left of $\Pref_{j,t}(y)$ and 
strictly to the right of $\Pref_{j,t}(y)$, respectively.
An example is illustrated in \Cref{fig:lem3-1}.

\begin{figure}[t]
\centering
\begin{tikzpicture}[
    scale=0.85,
    transform shape,
    vertex/.style={circle,draw,minimum size=3.5mm,inner sep=0pt},
    svertex/.style={circle,draw,minimum size=2mm,inner sep=0pt},
    labelnode/.style={draw=none,inner sep=1pt},
    every path/.style={thick}
]

\node[vertex] (y0) at (-1,0) {};
\node[labelnode,above=0.5mm of y0] {$y_0$};

\node[vertex] (y1) at (1.5,-1.0) {};
\node[labelnode,above right=0mm and 0mm of y1] {$y_1$};

\node[vertex] (y2) at (4.0,-2.0) {};
\node[labelnode,right=0.5mm of y2] {$y_2$};

\node[vertex] (yk1) at (4.0,-3.5) {};
\node[labelnode,above right=0mm and 0mm of yk1] {$y_{2^{j-i}-1}$};

\node[vertex] (yk) at (4.0,-4.5) {};
\node[labelnode,below=0.5mm of yk] {$y_{2^{j-i}}$};

\draw (y0)--(y1)--(y2);
\draw[dashed] (y2)--(yk1);
\draw (yk1)--(yk);

\node[vertex] (a01) at (-3.5,-1.0) {};
\node[svertex] (la01) at (-3.8,-6.0) {};
\node[svertex] (ra01) at (-3.2,-6.0) {};
\draw[dashed] (a01)--(la01);
\draw[dashed] (a01)--(ra01);
\node[draw=none] at ($(la01)!0.5!(ra01)$) {\scalebox{0.7}{$\cdots$}};
\node[vertex] (a02) at (-2.5,-1.0) {};
\node[svertex] (la02) at (-2.8,-6.0) {};
\node[svertex] (ra02) at (-2.2,-6.0) {};
\draw[dashed] (a02)--(la02);
\draw[dashed] (a02)--(ra02);
\node[draw=none] at ($(la02)!0.5!(ra02)$) {\scalebox{0.7}{$\cdots$}};
\node[vertex] (a03) at (10.5,-1.0) {};
\node[svertex] (la03) at (10.2,-6.0) {};
\node[svertex] (ra03) at (10.8,-6.0) {};
\draw[dashed] (a03)--(la03);
\draw[dashed] (a03)--(ra03);
\node[draw=none] at ($(la03)!0.5!(ra03)$) {\scalebox{0.7}{$\cdots$}};
\draw (y0)--(a01) (y0)--(a02) (y0)--(a03);
\node[draw=none,left=-0.5mm of la01] {$\bigl[$};
\node[draw=none,right=-0.5mm of ra02] {$\bigr]$};
\node[draw=none,left=-0.5mm of la03] {$\bigl[$};
\node[draw=none,right=-0.5mm of ra03] {$\bigr]$};

\node[vertex] (a11) at (-1,-2.0) {};
\node[svertex] (la11) at (-1.3,-6.0) {};
\node[svertex] (ra11) at (-0.7,-6.0) {};
\draw[dashed] (a11)--(la11);
\draw[dashed] (a11)--(ra11);
\node[draw=none] at ($(la11)!0.5!(ra11)$) {\scalebox{0.7}{$\cdots$}};
\node[vertex] (a12) at (0,-2.0) {};
\node[svertex] (la12) at (-0.3,-6.0) {};
\node[svertex] (ra12) at (0.3,-6.0) {};
\draw[dashed] (a12)--(la12);
\draw[dashed] (a12)--(ra12);
\node[draw=none] at ($(la12)!0.5!(ra12)$) {\scalebox{0.7}{$\cdots$}};
\node[vertex] (a13) at (9,-2.0) {};
\node[svertex] (la13) at (8.7,-6.0) {};
\node[svertex] (ra13) at (9.3,-6.0) {};
\draw[dashed] (a13)--(la13);
\draw[dashed] (a13)--(ra13);
\node[draw=none] at ($(la13)!0.5!(ra13)$) {\scalebox{0.7}{$\cdots$}};
\draw (y1)--(a11) (y1)--(a12) (y1)--(a13);
\node[draw=none,left=-0.5mm of la11] {$\bigl[$};
\node[draw=none,right=-0.5mm of ra12] {$\bigr]$};
\node[draw=none,left=-0.5mm of la13] {$\bigl[$};
\node[draw=none,right=-0.5mm of ra13] {$\bigr]$};

\node[vertex] (b01) at (2.5,-4.5) {};
\node [svertex] (lb01) at (2.2,-6.0) {};
\node [svertex] (rb01) at (2.8,-6.0) {};
\draw[dashed] (b01)--(lb01);
\draw[dashed] (b01)--(rb01);
\node[draw=none] at ($(lb01)!0.5!(rb01)$) {\scalebox{0.7}{$\cdots$}};
\node[vertex] (b02) at (5.5,-4.5) {};
\node [svertex] (lb02) at (5.2,-6.0) {};
\node [svertex] (rb02) at (5.8,-6.0) {};
\draw[dashed] (b02)--(lb02);
\draw[dashed] (b02)--(rb02);
\node[draw=none] at ($(lb02)!0.5!(rb02)$) {\scalebox{0.7}{$\cdots$}};
\node[vertex] (b03) at (6.5,-4.5) {};
\node [svertex] (lb03) at (6.2,-6.0) {};
\node [svertex] (rb03) at (6.8,-6.0) {};
\draw[dashed] (b03)--(lb03);
\draw[dashed] (b03)--(rb03);
\node[draw=none] at ($(lb03)!0.5!(rb03)$) {\scalebox{0.7}{$\cdots$}};
\draw (yk1)--(b01) (yk1)--(b02) (yk1)--(b03);
\node[draw=none,left=-0.5mm of lb01] {$\bigl[$};
\node[draw=none,right=-0.5mm of rb01] {$\bigr]$};
\node[draw=none,left=-0.5mm of lb02] {$\bigl[$};
\node[draw=none,right=-0.5mm of rb03] {$\bigr]$};

\node[draw=none,left=1mm of la01] {$\Bigl\{$};
\node[draw=none] at ($(ra02)!0.5!(la11)$) {$,$};
\node[draw=none] at ($(ra12)!0.5!(lb01)$) {$,\ \cdots\ ,$};
\node[draw=none,right=1mm of rb01] {$\Bigr\}$};
\node[draw=none] at ($(la01)!0.5!(rb01)+(0,-6mm)$) {$=\LeftIntervals_t(y)$};

\node[draw=none,left=1mm of lb02] {$\Bigl\{$};
\node[draw=none] at ($(rb03)!0.5!(la13)$) {$,\ \cdots\ ,$};
\node[draw=none] at ($(ra13)!0.5!(la03)$) {$,$};
\node[draw=none,right=1mm of ra03] {$\Bigr\}$};
\node[draw=none] at ($(lb02)!0.5!(ra03)+(0,-6mm)$) {$=\RightIntervals_t(y)$};

\end{tikzpicture}
\caption{Illustration of the interval families $\LeftIntervals_t(y)$ and $\RightIntervals_t(y)$. $y=y_0$ is the parent of $y_{2^{j-i}}$ in $R_j$, while $y_0\to y_1\to\cdots\to y_{2^{j-i}}$ is a path in $R_i$.}
\label{fig:lem3-1}
\end{figure}
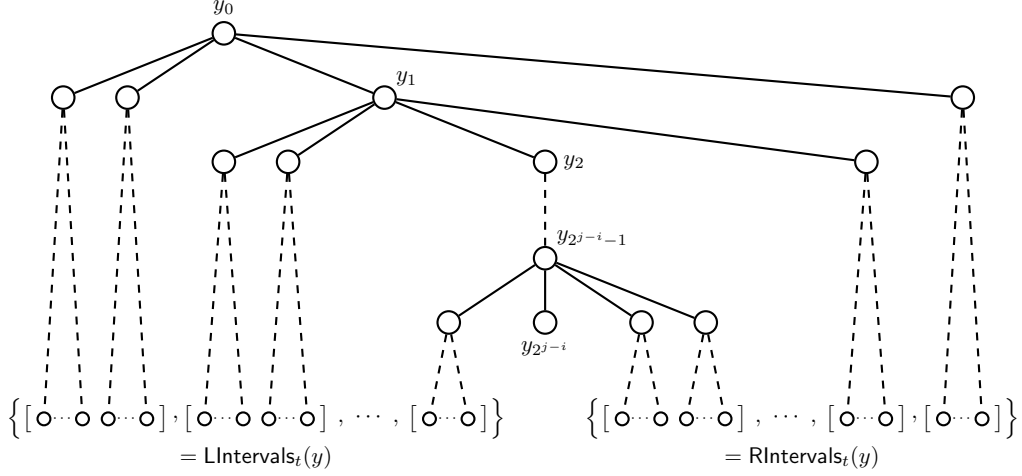

The potential is defined as

$$\begin{aligned}
    \Phi_t(y) &=|\Visible_t(\LeftIntervals_t(y))|+|\Visible_t(\RightIntervals_t(y))|,\\
    \Phi_t &=\sum_{y\in V(R_j),\ y \text{ is not a leaf}}\Phi_t(y).
\end{aligned}$$

We now consider how the potential changes in 
response to an access $\sigma(t)$.  Observe that
from $y$'s point of view, the simplest scenario
is where $\sigma(t)$ is either not in $\Descendants_k(y)$,
or in $\Descendants_k(\Pref_{j,t-1}(y))$, in which case
$\Pref_{j,t}(y)=\Pref_{j,t-1}(y)$ does not change,
and the change to $\Phi_t(y)$ is well behaved, as specified in \Cref{prop:hull}.
If this is not the case, then 
$\Pref_{j,t}(y)\neq\Pref_{j,t-1}(y)$
and $\Phi_t(y)$ may increase by the maximum 
$\Omega(2^{j-i})$.  This increase will be charged
to a unit of the interleave bound $\Lambda(\sigma,R_j,t)$.

We shall analyze $\Phi_t-\Phi_{t-1}$ in two steps, one in which we change preferred children, and one in which we consider the change to the visibility sets.

\begin{enumerate}
    \item First change the preferred children of the ancestors of $\sigma(t)$.

    For $y\in \Ancestors_j(\sigma(t))$, its potential may change only when $\Pref_{j,t}(y)\ne \Pref_{j,t-1}(y)$. This happens for $\Lambda(\sigma,R_j,t)$ different nodes $y$. Since the potential of a single node $y$ is always upper bounded by $2^{j-i+1}$---the maximum cardinality of $\LeftIntervals_t(y)\cup\RightIntervals_t(y)$---the total increase in potential is at most $\Lambda(\sigma,R_j,t)\cdot 2^{j-i+1}$.

    \item Next add several new points in the geometric view and measure the change in visibility sets.

    At this point the path from any reference tree root to $\sigma(t)$ is already preferred. As a consequence, we have $\sigma(t)\notin \bigcup(\LeftIntervals_t(y)\cup \RightIntervals_t(y))$ for any node $y$ in $R_j$. 
    From the second case of \Cref{prop:hull}, 
    the potential of $y$ can only decrease.

    Let $y\in V(R_j)$ be an ancestor of $\sigma(t)$,
    and suppose $q^L(y)$ and $q^R(y)$ are the number of intervals from $\LeftIntervals_t(y)$ and $\RightIntervals_t(y)$ that were touched when 
    accessing $\sigma(t)$. 
    We can rewrite $\GCost_i(t)$ and $\GCost_j(t)$ as follows.

    $$\begin{aligned}
        \GCost_i(t)&=\sum_{y\in\Ancestors_j(\sigma(t))}\left(q^L(y)+q^R(y)\right),\\
        \GCost_j(t)&=\sum_{y\in\Ancestors_j(\sigma(t))}\left(\ind[q^L(y)\ge 1]+\ind[q^R(y)\ge 1]\right).
    \end{aligned}$$

    From the second case of \Cref{prop:hull}, 
    the decrease in $y$'s potential 
    is lower bounded by $\max\left\{0,q^L(y_0)-2\right\}+\max\left\{0,q^R(y_0)-2\right\}$. Thus, the total decrease of potential is at least $\GCost_i(t)-2\cdot \GCost_j(t)$.
\end{enumerate}

Now we have

$$\Phi_t\le\Phi_{t-1}+2^{j-i+1}\cdot \Lambda(\sigma, R_j, t)-\left(\GCost_i(t)-2\cdot \GCost_j(t)\right).$$

By telescoping over $t$ and rearranging, we have

$$\GCost_i\le 2\cdot \GCost_j+2^{j-i+1}\cdot\IB(\sigma,R_j)+\left(\Phi_0-\Phi_m\right).$$

Finally, we observe that $\Phi_0\le 2^{j-i+1}\cdot|R_j|\le 2^{j-i+2}\cdot n$.  This completes the proof of \Cref{lem:v}.

\subsection{Proof of \Cref{lem:w}}\label{sec:l2}

Fix reference trees $R_i$ and $R_j$ with $i<j$.  We will again use a potential argument to prove the lemma, similar to the one used in \Cref{sec:l1}.

For a node $z$ in $R_i$ but not in $R_j$, let

$$\Intervals(z) = \left\{\Descendants_k(z')\mid z'\in\Descendants_j(z)\right\}$$

be the set of intervals represented by nodes in $\Descendants_j(z)$.

We impose a potential $\Psi_t$ defined as

$$\begin{aligned}
    \Psi_t(z)&:=|\Visible_t(\Intervals(z))\setminus \{\Descendants_k(\Pref_{j,t}(z))\}|,\\
    \Psi_t&:=\sum_{z\in V(R_i)\setminus V(R_j)}\Psi_t(z),
\end{aligned}$$

We now analyze the change in potential from time $t-1$ to $t$, using the same steps from \Cref{lem:v}.  First we change the preferred children to be consistent with the access $\sigma(t)$, then account for the change in visibility caused by adding points to arboreally satisfy $\sigma(t)$.

\begin{enumerate}
    \item First change the preferred children of the ancestors of $\sigma(t)$.

    Suppose the preferred child of some node $z_0$ changes. Find the nearest ancestor of $z_0$ that is an $R_j$-node, and let it be $z_1$.  This switch will make $\Pref_{j,t}(z)\ne\Pref_{j,t-1}(z)$ for all $z$ on the path from $z_1$ to $z_0$. There are as many as 
    $2^{j-i}$ such nodes and the potential of each such node increases by at most $1$. 
    Thus, the total increase in potential is upper bounded by $\Lambda(\sigma,R_i,t)\cdot 2^{j-i}$.

    \item Second, add new points to the geometric view.

    We categorize the nodes $z\in R_i$ into three types 
    and bound the change in their potential separately.

    \begin{description}
        \item[Case: $z$ is an ancestor of $\sigma(t)$.]

        From the first case of \Cref{prop:hull}, we know the number of visible intervals in $\Intervals(z)$ increases by at most $2$. Since we exclude $\Pref_{j,t}(z)$ (which is also an ancestor of $\sigma(t)$) from the intervals counted by $\Psi_t(z)$, the potential can only increase when some interval in $\Intervals(z)$ other than $\Descendants_k(\Pref_{j,t}(z))$ is touched. 
        We say $z$ is \emph{special} in this case.

\medskip

Let $y_0\to y_1\to\cdots\to y_{2^{j-i}}$ be a preferred path in $R_i$ between two $R_j$ nodes $y_0,y_{2^{j-i}}$ that are ancestors of $\sigma(t)$, i.e. $y_0,y_{2^{j-i}}\in\Ancestors_j(\sigma(t))$. For any $z=y_q$, $1\le q\le 2^{j-i}$, we have $\Pref_{j,t}(z)=y_{2^{j-i}}$. Thus, $z$ is special  when some sibling of $y_{2^{j-i}}$ in $R_j$ is touched, which is precisely what is tallied by $\GCost_j(t)$.

\medskip 

        Thus, the number of special nodes associated with $\sigma(t)$ is bounded by $2^{j-i}\cdot\GCost_j(t)$, and the total increase in potential for all such nodes is at most $2^{j-i+1}\cdot\GCost_j(t)$.

        \item[Case: The parent of $z$, but not $z$, is an ancestor of $\sigma(t)$.]
        Define 
        \[
Z(t) =\bigcup_{z^*\in\Ancestors_i(\sigma(t))}\Siblings_i(z^*)
\]
to be the set of all such nodes.        
        From the second case of \Cref{prop:hull} we know the potential of $z$ does not increase. Supposing that $q(z)$ intervals in $\Intervals(z)$ are touched, we can rewrite $\ICost_i(t)$ and $\ICost_j(t)$ as
    
        $$\begin{aligned}
            \ICost_i(t)&=\sum_{z\in Z(t)}\ind[q(z)\ge 1],\\
            \ICost_j(t)&=\sum_{z\in Z(t)}q(z).
        \end{aligned}$$
    
        The decrease of $|\Visible(\Intervals(z))|$ is lower bounded by $\max\{0,q(z)-2\}$. Thus, the decrease of the potential of $z$ is lower bounded by $\max\{0,q(z)-2\}$ if $\Descendants_k(\Pref_{j,t}(z))$ does not become invisible, or $\max\{0,q(z)-3\}$ if $\Descendants_k(\Pref_{j,t}(z))$ becomes invisible.\footnote{Note that this argument holds for $z\in V(R_j)$ as well, although it does not have a potential, since 
        $q(z)\le|\Intervals(z)|=1$ in this case and thus $\max\{0,q(z)-2\}=0$.}

\medskip

        Once $\Descendants_k(\Pref_{j,t}(z))$ becomes invisible, it cannot become visible again until some element in the subtree rooted at $z$ is accessed, which means the preferred child of the parent of $z$ changes to $z$. Thus, denote $\InvisiblePref(t)$ as the number of $z\in Z(t)$ such that $\Descendants_k(\Pref_{j,t}(z))$ is visible before this access and becomes invisible after this access. We have $\sum_t\InvisiblePref(t)\le\IB(\sigma,R_i)$.

        \medskip 
        
        The total decrease in potential for all nodes in $Z(t)$ is at least $\ICost_j(t)-2\cdot \ICost_i(t)-\InvisiblePref(t)$
    
        \item[Case: Neither $z$ nor the parent of $z$ is an ancestor of $\sigma(t)$.]
    
        From the second case of \Cref{prop:hull} we know the potential of $z$ does not increase.
    \end{description}
\end{enumerate}

Summing over all the potential changes above, we obtain

$$\Psi_t\le\Psi_{t-1}+2^{j-i}\cdot \Lambda(\sigma,R_i,t)+2^{j-i+1}\cdot \GCost_j(t)-(\ICost_j(t)-2\cdot \ICost_i(t)-\InvisiblePref(t)).$$

By telescoping over $t$ and rearranging, we have

$$\ICost_j\le 2\cdot \ICost_i+\IB(\sigma,R_i)+2^{j-i}\cdot\left(\IB(\sigma, R_i)+2\cdot \GCost_j\right)+(\Psi_0-\Psi_m).$$

Each node in $R_j$ occurs in $\Intervals(z)$ for exactly those nodes $z\in R_i$ which are ancestors at distance less than $2^{j-i}$.  Since $R_j$ has less than $2n$ nodes, $\Psi_0 < 2^{j-i+1}\cdot n$. This completes the proof of \Cref{lem:w}.

\subsection{Proof of \Cref{thm:main}}\label{sec:main_thm}

Recall that we have conveniently chosen $k$ to be a 
square and defined $p=\sqrt{k}$.  First consider the cost measures $\GCost_{ip}$ for $0\leq i\leq p$.  By \Cref{lem:v} and the fact that $\IB(\sigma,R_{(i+1)p})\le\IB(\sigma,R_0)$,
\begin{align*}
\GCost_{ip} &\le 2\cdot \GCost_{(i+1)p}+2^{p}\cdot(2\cdot \IB(\sigma, R_0)+4n),\\
\intertext{and by combining these inequalities for all $i \in [0,p-1]$, we have}
\GCost_0 &\le 2^p\cdot \GCost_k+2^{2p}\cdot(2\cdot \IB(\sigma,R_0)+4n).\\
\intertext{Since $\GCost_k\le 2m$, we have}
\ICost_0 = \GCost_0 &\leq 
            O(2^{2p})\cdot(\IB(\sigma,R_0)+n+m).\\
\intertext{Now, consider the cost measures $\ICost_{ip}$ for $0\leq i\leq p$.  By \Cref{lem:w} and $\IB(\sigma,R_{ip})\le\IB(\sigma,R_0)$, we have}
\ICost_{(i+1)p} &\le 2\cdot \ICost_{ip}+2^{p}\cdot\left(\IB(\sigma,R_0)+2n+2\cdot\GCost_{(i+1)p}\right)\\
    &\le 2\cdot \ICost_{ip}+2^{p}\cdot\left(\IB(\sigma,R_0)+2n+2\cdot \GCost_{0}\right)\\
    &\le 2\cdot \ICost_{ip}+
        O\left( 2^{3p}\right)\cdot(\IB(\sigma,R_0)+n+m),
\intertext{and by combining these inequalities for all $i\in [0,p-1]$,}
\ICost_k &\le 2^p\cdot \ICost_0+O\left(2^{4p}\right)\cdot(\IB(\sigma,R_0)+n+m).
\intertext{Finally, note that $\ICost_0=\GCost_0=O\left(2^{2p}\right)\cdot(\IB(\sigma,R_0)+n+m)$, thus}
\ICost_k &=O\left(2^{4p}\right)\cdot(\IB(\sigma,R_0)+n+m)=2^{O\left(\sqrt {\log\log n}\right)}\cdot(\IB(\sigma,R_0)+n+m).
\end{align*}

A more refined analysis yields a smaller constant $2\sqrt 2$ in the exponent. 
See Appendix~\ref{sec:improvement}.

\section{Conclusion}\label{sect:conclusion}

In this paper we proved the first non-trivial competitiveness bound on one of the ``organic'' candidates for dynamic optimality, namely \Greedy.  The most important contribution of this work is the \emph{scaling} approach to cost analysis in which we define a series of progressively refined cost measures and analyze the relationship between successive measures in a chain.

Our main result is a competitiveness bound of $2^{O(\sqrt{\log\log n})}$, which is surely not tight.  
It is an artifact of the leading 
coefficient in \Cref{lem:v,lem:w} being \emph{2}.
If both \emph{2}s could be improved to \emph{1}s, 
this would immediately result in a $\text{poly}(\log\log n)$-competitiveness bound.

\bibliographystyle{alpha}
\bibliography{bibliography}

\appendix

\section{An Improvement of the Constant in the Exponent}
\label{sec:improvement}

Assume without loss of generality that $n=2^{2^k}$ for some $k=p(p+1)/2$ and integer $p=\sqrt{2}\cdot\sqrt{\log\log n}+O(1)$.

First consider the cost measures $\GCost_{p(p+1)/2-i(i+1)/2}$ for $0\le i\le p$. By \Cref{lem:v},

$$\GCost_{p(p+1)/2-i(i+1)/2}\le 2\cdot \GCost_{p(p+1)/2-i(i-1)/2}+2^{i}\cdot(2\cdot \IB(\sigma, R_0)+4n),$$

and by repeatedly expanding this inequality, we derive

$$\GCost_0\le 2^p\cdot \GCost_k+p\cdot 2^{p}\cdot(2\cdot \IB(\sigma,R_0)+4n).$$

Since $\GCost_k\le 2m$, we have

$$\ICost_0=\GCost_0\leq 4p\cdot 2^{p}\cdot(\IB(\sigma,R_0)+n+m).$$

Now, consider the cost measures $\ICost_{i(i+1)/2}$ for $0\leq i\leq p$.  By \Cref{lem:w}, we have

$$\begin{aligned}
    \ICost_{i(i+1)/2}&\le 2\cdot \ICost_{i(i-1)/2}+2^{i}\cdot\left(\IB(\sigma,R_0)+2n+2\cdot\GCost_{i(i+1)/2}\right)\\
    &\le 2\cdot \ICost_{i(i-1)/2}+2^{i}\cdot\left(\IB(\sigma,R_0)+2n+2\cdot \GCost_{0}\right)\\
    &\le 2\cdot \ICost_{i(i-1)/2}+10p\cdot 2^{p+i}\cdot(\IB(\sigma,R_0)+n+m).
\end{aligned}$$

and by repeatedly expanding this inequalities, we derive

$$\ICost_k\le 2^p\cdot \ICost_0+10p\cdot 2^{2p}\cdot(\IB(\sigma,R_0)+n+m).$$

Finally, note that $\ICost_0=\GCost_0\leq 4p\cdot 2^{p}\cdot(\IB(\sigma,R_0)+n+m)$, so we have

$$\ICost_k\le 14p\cdot 2^{2p}\cdot(\IB(\sigma,R_0)+n+m).$$

Recall that $p=\sqrt{2}\cdot\sqrt{\log\log n}+O(1)$, thus $\Greedy$ is $2^{(2\sqrt 2+o(1))\sqrt{\log\log n}}$-competitive when $m\ge n$.

\end{document}